\documentclass[aps,preprint,superscriptaddress,nofootinbib]{revtex4-2}

\usepackage[T1]{fontenc}
\usepackage[utf8]{inputenc}
\usepackage{amsmath,amssymb,amsfonts,amsthm,mathtools}
\usepackage{graphicx}
\usepackage{booktabs}
\usepackage{multirow}
\usepackage{hyperref}
\usepackage{epstopdf}

\newtheorem{theorem}{Theorem}

\theoremstyle{definition}

\theoremstyle{remark}

\begin{document}

\title{The thermodynamic efficiency of coupled chaotic dissipative structures}

\author{\'Alvaro G. L\'opez}
\email{alvaro.lopez@urjc.es}
\affiliation{Nonlinear Dynamics, Chaos and Complex Systems Group, Department of Geology, Physics and Inorganic Chemistry, Universidad Rey Juan Carlos, 28933 M\'ostoles, Madrid, Spain}

\author{In\'es P. Mari\~no}
\affiliation{Nonlinear Dynamics, Chaos and Complex Systems Group, Department of Geology, Physics and Inorganic Chemistry, Universidad Rey Juan Carlos, 28933 M\'ostoles, Madrid, Spain}

\author{Alfonso Delgado-Bonal}
\affiliation{University of Maryland Baltimore County, Baltimore, MD, United States}
\affiliation{Climate and Radiation Laboratory, NASA Goddard Space Flight Center, Greenbelt, MD, United States}

\date{\today}

\begin{abstract}
Dissipative structures are open dynamical systems that sustain coherent macroscopic organization by continuously exchanging energy and matter with their environment and generating entropy. A recent thermodynamic analysis of the paradigmatic Malkus--Lorenz waterwheel interpreted the Lorenz system as an engine, deriving an exact formula for its thermodynamic efficiency, and showing that efficiency tends to increase as the system is driven far from equilibrium while displaying sharp drops near the Hopf subcritical bifurcation to chaos. Here, we extend that single-engine framework to coupled dissipative structures. We introduce two canonical couplings---master-slave coupling (series) and symmetric diffusive coupling (parallel)---and prove two fundamental association laws allowing us to reduce the composite systems to an equivalent engine with a specified efficiency. We then apply these abstract results to coupled Lorenz waterwheels, deriving efficiency formulas consistent with the underlying power balance. We perform numerical simulations confirming that (a) series coupling induces an increase in thermodynamic efficiency, (b) parallel coupling averages the efficiency of engines and increases total energy flow, (c) synchronization is typically neutral or beneficial for efficiency except in narrow parameter regions, and (d) coupling modifies the curvature of entropy-generation trends. Our theorems suggest a mathematically rigorous and transparent route to define and compute thermodynamic efficiency for generalized flow networks, with potential application to complex systems energetics.
\end{abstract}


\maketitle

\section{Introduction}

The concept of \emph{dissipative structure} describes open systems that operate far from equilibrium by sustaining continuous flows of matter and energy and producing entropy \cite{Prigogine1968,Prigogine1978}. Across physics, chemistry and biology, such systems self-organize into coherent macroscopic cycles and oscillations, frequently through universal bifurcations that break fundamental symmetries. As these systems are driven further from equilibrium, they typically destabilize and experience chaotic dynamics, restoring such symmetries. Generally, dissipative structures are locally active systems \cite{MainzerChua2013}, and by symmetry-breaking they can amplify environmental energy fluctuations and convert them into coherent directed motion. In this regard, dissipative structures can be classified as active matter systems \cite{Vala24}. In the Malkus--Lorenz waterwheel, these ideas crystallize into a minimal and mathematically tractable engine whose reduced dynamics is equivalent (up to affine transformations) to the Lorenz system \cite{Strogatz2015}. 

A recent thermodynamic analysis \cite{Lopez2023Lorenz} interpreted the Malkus--Lorenz wheel as a genuine self-excited \emph{engine} \cite{cor36, jen13} operating between two reservoirs of matter (high-potential and low-potential energy water streams), derived an exact analytic expression for its thermodynamic efficiency as a function of phase-space variables and parameters. It showed a general tendency toward increased efficiency as driving increases, punctuated by abrupt drops at critical bifurcations leading to chaotic dynamics. In this simple system, thermodynamic efficiency is related to entropy generation \cite{Lopez2023Lorenz}, enabling a direct link between exergy degradation and thermodynamic irreversibility. The effects of energy rate-loss of matter on thermodynamic efficiency by producing negentropy were also addressed in previous works, while the stationary states of dissipative structures were connected with the physiological concept of homeostasis \cite{Lopez2023Lorenz}. 

Such analysis allowed for conceptualizing Boltzmann's suggestion of life’s ``struggle for existence'' as a struggle over entropy. Lotka further stated that natural selection favors organisms that are better at capturing available external energy and channeling it into processes that support survival \cite{Lotka1922}. Later on, Odum renamed this conceptual framework as the maximum power principle: selection should favor systems that increase overall energy flow, so that both efficiency and total power would be relevant for ecosystem development \cite{Odum1995}. This efficiency reflects how effectively a dissipative structure captures incoming resources, performs work, exports energy to its surroundings, and manages internal dissipation. However, it is not clear the extent to which these two factors contribute to Darwinian evolution, and how they materialize in the evolution of dissipative structures.

The present work aims to provide insight into this principle by moving from a single thermodynamic engine to \emph{coupled assemblies} and, ultimately, to \emph{general energy-flow networks} in complex systems, with potential application to ecosystem energetics \cite{SchneiderKay1994}. Ecosystems are comprised by intricate networks of trophic chains, which can be interpreted as a directed cascade of coupled thermodynamic engines, in the sense used throughout this work: each level receives an input exergy power flux, converts a fraction into ``useful'' functional work (biomass production, motion, reproduction, ecological work), and rejects the remaining fraction as degraded energy and waste streams (heat, respiration losses, detritus). 

In this interpretation, primary producers constitute a first stage (capturing an external gradient, e.g., radiative or chemical inflows), herbivores and successive predator levels form subsequent stages that are \emph{in series} because part of the rejected or available flux constrains their effective input from the previous level \cite{Fath2007ENA}. In parallel, decomposers constitute an additional branch that draws from the detrital pool generated across multiple trophic levels; thus, decomposers can be modeled as a \emph{parallel} engine coupled to the series chain of engines and a pump, in the sense that they share the overall ecosystem input indirectly while partitioning the total available flux into an additional pathway. Conceptually, one may therefore view this simplified scheme as an idealized ecosystem module formed by ``one pump'', ``four wheels in series'' (from producers to top predators) coupled to ``one wheel in parallel'' (decomposers), enabling the direct use of series--parallel association laws to define global efficiency and to separate topology-imposed constraints from dynamical flow selection \cite{Lindeman1942}.

Motivated by trophic chains and ecosystem energetics, which here are only used as an analogy, we study two canonical elementary coupling topologies: (i) directed \emph{series} coupling, in which the outflow of an upstream element feeds a downstream element forming a cascade; and (ii) \emph{parallel} coupling with diffusive exchange, in which similar engines are coupled to the same environment matter and energy reservoirs, and also symmetrically through exchange currents among them. The guiding question is whether thermodynamic efficiency can be meaningfully extended from a single dissipative structure to multi-node networks and whether topology yields provable constraints and composition rules for efficiency.

The present work is organized as follows. In Sec.~2 we introduce the Malkus-Lorenz waterwheel, briefly explain its fundamental thermodynamic properties as an engine, and the effects of chaotic dynamics on thermodynamic efficiency. Then, in Sec.~3, we arrange two waterwheels, showing how these systems can be physically coupled in series or in parallel, integrating more sophisticated systems of engines. We then proceed to prove two fundamental theorems that allow us to compute the net thermodynamic efficiency of the reduced equivalent system. Next, we apply these general laws to the computation of the thermodynamic efficiency of coupled Lorenz systems, illustrating their validity of the general laws through particular examples. Finally, we compute the entropy generation of coupled systems and explain from the theoretical perspective the modifications introduced in entropy-generation bifurcation curves as the external driving from the energy source is increased. 

\section{Thermodynamic efficiency and entropy generation in the Malkus--Lorenz engine}

We now introduce the fundamental self-excited dynamical system that is used to exemplify the effect of coupling dissipative units on the thermodynamic efficiency: the Malkus-Lorenz waterwheel \cite{mis06,kim17}. By studying minimally complex arrangements of water wheels, we pave the way towards the understanding of thermodynamic efficiency and entropy flows in more complex networks. 

\subsection{Malkus--Lorenz waterwheel and reduced dynamics}

The Malkus--Lorenz wheel is an open dissipative dynamical system that is fed with a continuous inflow of water $Q(\theta)$ at high potential energy, and which sustains cyclic motion by releasing part of the water income at a fractional rate $K$, thus degrading an external energy gradient. In practice, the wheel consists of compartments full of water on a tilted rotating ring, where $\theta$ denotes the position of each compartment along such ring. The water mass distribution can be approximated by specifying the linear mass density $\mu(\theta,t)$ along the ring. The dynamics of this extended system is given by the continuity equation ensuring mass conservation, and Newton's second law promoting the rotational motion of the wheel. The former reads
\begin{equation}
\dfrac{\partial \mu(\theta,t)}{\partial t}=Q(\theta)-K \mu(\theta,t)-\dfrac{\partial}{\partial \theta}\mu(\theta,t) \omega(t),
\label{eq:1}
\end{equation}
where $\omega(t)$ is the angular speed of the wheel, which, being solid, moves as a whole. The latter can be expressed by means of Euler's second equation, describing the rotation of the wheel, which yields the integro-differential equation
\begin{equation}
I \dot{\omega}(t)=-\nu \omega(t)+g r \sin\alpha \int_0^{2 \pi} \mu(\theta,t)\sin\theta d\theta,
\label{eq:2}
\end{equation}
where $\nu$ represented the friction from the magnetic brake, which ultimately flows into the surroundings in the form of heat. We have assumed that the moment of inertia of the wheel $I$ is fixed, which holds once the transient dynamics has exhausted. Under the standard Fourier expansion of the mass distribution along the rim of the wheel, we write
\begin{equation}
\mu(\theta,t)=\sum_{n=0}^\infty(a_n(t)\sin(n\theta)+b_n(t)\cos(n\theta)).
\label{eq:3}
\end{equation}
Assuming that the inflow of water is symmetrical, we can expand the income rate function $Q(\theta)$ as an even Fourier series with coefficients $q_n$. The resulting set of equations can be written for all the modes of the wheel as
\begin{equation}
    \begin{array}{lr}
       \dot{a}_n(t)= n \omega(t) b_n(t)-K a_n(t) \bigskip\\
       \dot{b}_n(t)= -n \omega(t) a_n(t)-K b_n(t) + q_n  \bigskip\\
       \dot{\omega}(t)=-\dfrac{\nu}{I} \omega(t)+\dfrac{\pi g r \sin\alpha}{I} a_1(t)
     \end{array}.
     \label{eq:4}
\end{equation}
Note that, because the torque produced by gravity changes harmonically along the ring of the wheel, only the first-order terms of the Fourier expansion of the mass distribution couple to the rotational motion. This yields a three-dimensional ODE system comprised of the equations
\begin{equation}
    \begin{array}{lr}
       \dot{a}_1(t)=  \omega(t) b_1(t)-K a_1(t)   \bigskip\\
       \dot{b}_1(t)= -\omega(t) a_1(t)-K b_1(t) + q_1  \bigskip\\
       \dot{\omega}(t)=-\dfrac{\nu}{I} \omega(t)+\dfrac{\pi g r \sin\alpha}{I} a_1(t)
     \end{array}.
     \label{eq:5}
\end{equation}
It can be mathematically demonstrated \cite{Strogatz2015} that this system is related to the Lorenz equations after an affine transformation, yielding the system of differential equations
\begin{equation}
    \begin{array}{lr}
       \dot{x}(t)= \sigma(y(t)-x(t)) \bigskip\\
       \dot{y}(t)= \rho x(t)-x(t)z(t)-y(t) \bigskip\\
       \dot{z}(t)= x(t)y(t)-z(t) 
     \end{array},
          \label{eq:6}
\end{equation}
where the variable $x(t)$ is related to the angular velocity of rotation, the variable $y(t)$  relates to $a_1(t)$, and $z(t)$ can be expressed in terms of $b_1(t)$. Furthermore, the non-dimensional Rayleigh $\rho=\pi g r q_1/\nu K^2$ and Prandtl $\sigma=\nu/I K$ numbers have been introduced. We highlight the importance of the first of these two parameters. As we shall see, the thermodynamic efficiency of the wheel can be expressed as a function of this parameter alone when it is far from equilibrium but in the regular regime, before the onset of chaos. 
\begin{figure}
\centering
\includegraphics[width=0.62\linewidth]{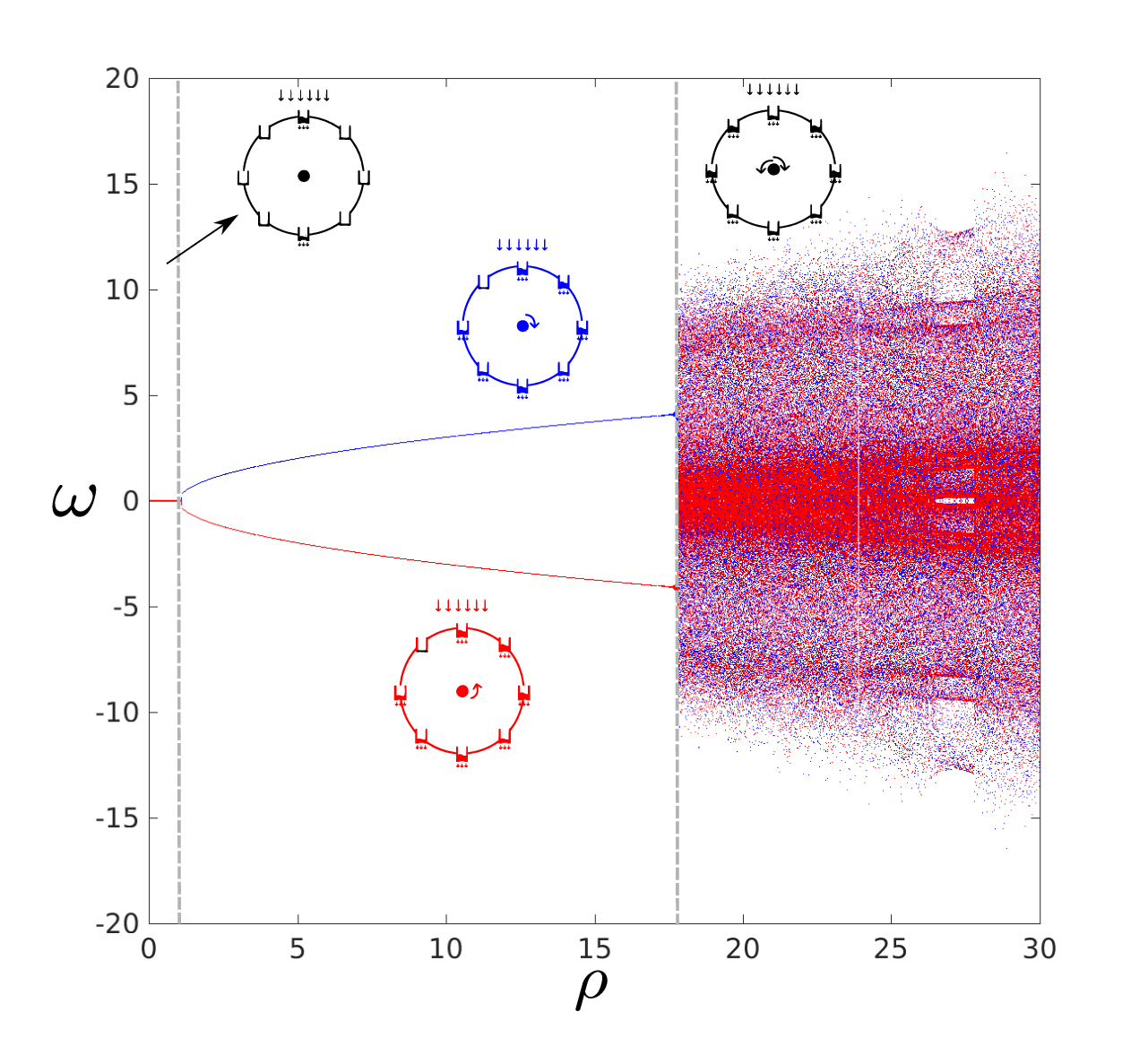}
\caption{The values of $\omega(t)$ at the various limit sets of the system as the Rayleigh parameter $\rho$ is varied. We distinguish three regimes. Firstly, a wheel at rest for $\rho < 1$ can be seen. Secondly, a supercritical Pitchfork bifurcation befalls at $\rho = 1$, and we find a wheel spinning uniformly clockwise or anticlockwise, depending on the arm of the diagram selected. Finally, a subcritical Hopf bifurcation takes place
for $\rho_c = 17.5$. The motion of the wheel becomes chaotic, alternating unpredictably from a clockwise to an anticlockwise rotation.}
\label{fig:fig1}
\end{figure}

This system presents three basic dynamical regimes (see Fig.~\ref{fig:fig1}). In the first regime (i), when the water inflow from the hose is too small ($\rho<1$) the system is passive, and a stationary flow of water transverses the wheel, which is at rest. Then, as the water income increases, a supercritical Pitchfork bifurcation takes place for $\rho=1$. In this second regime (ii) the wheel sets to uniform rotational motion, by choosing one of the two symmetrical branches (clockwise or anti-clockwise). This constitutes a process of spontaneous symmetry breaking of the cyclic group $C_2$ of the wheel. As the wheel is increasingly feed at a faster rate, a second subcritical Hopf bifurcation takes place for $\rho_c=\sigma(\sigma+4)/(\sigma-2)$. The resulting chaotic dynamics restores the symmetry of the system, and the phase space trajectory depicts the well-known strange Lorenz attractor. Now the wheel alternatively and unpredictably switches between the two modes of rotation.

\subsection{Efficiency as a dynamical variable}

The wheel ingests water with high gravitational potential and expels water at lower potential. Let $\dot E_p^{(\mathrm{in})}$ and $\dot E_p^{(\mathrm{out})}$ be the rates of incoming and outgoing gravitational potential power. According to the first law of thermodynamics, the energy balance \cite{Lopez2023Lorenz} gives $\dot W = \dot E_p^{(\mathrm{in})} - \dot E_p^{(\mathrm{out})} $ and the (time-dependent) efficiency
\begin{equation}
\eta(t) = \frac{\dot W}{\dot E_p^{(\mathrm{in})}} = 1 - \frac{\dot E_p^{(\mathrm{out})}}{\dot E_p^{(\mathrm{in})}}.
\label{eq:eta_def_general}
\end{equation}
This definition of work involves both the work done by the torque produced by gravity, as well as the potential energy of the wheel $\dot{W}=\dot{W}_{\tau}+\dot{E}_{p}$. Defined in this manner, we ensure that the thermodynamic efficiency is instantaneously between zero and one. Nevertheless, since we average the power in all our simulations over the attractor, in practice we have $\langle \dot{W} \rangle= \langle \dot{W}_{\tau} \rangle$, because $\langle \dot{E}_{p} \rangle = 0 $. For symmetric injection and leakage, one obtains the exact expression
\begin{equation}
\eta(t)=1-\frac{\pi K\,(2b_0(t)+b_1(t))}{2q(\beta+\sin\beta)},
\label{eq:eta_single_exact}
\end{equation}
where only the first even and odd Fourier coefficients contribute to $\dot E_p^{(\mathrm{out})}$ \cite{Lopez2023Lorenz}. In the second region representing the regime (ii) shown in the bifurcation diagram, where the rotation of the wheel is uniform, we can solve Eq.~(4) using the first two Fourier coefficients of $Q(\theta)$, where the water inflow is assumed to be constant at a rate $q$ between the angles $(-\beta,\beta)$, where the hose is placed. This yields the coefficients $q_0=\beta q/\pi$ and $q_1=2q\sin\beta/\pi$, which allow to obtain from the fixed point analysis the value of the efficiency in closed analytical form
\begin{equation}
\eta^{*}=\frac{\sin\beta}{\beta+\sin\beta}\left(1-\frac{1}{\rho}\right).
\end{equation}
In our case, considering the limit $\beta \rightarrow 0$ and renormalizing the thermodynamic efficiency to its maximum value $\eta^{*}_{max}=1/2$, yields the neat expression
\begin{equation*}
\eta_{II}^{*}=1-\frac{1}{\rho},
\end{equation*}
a formula that resembles the efficiency of Carnot's cycle, where the quotient of temperatures between the hot and cold reservoirs has been replaced by the quotient of income and outcome rates of water flow. These results have been recently corroborated in studies of fluid convection \cite{Li23}, from which the Lorenz model was originally derived \cite{Lorenz1963}.

In the case of chaotic dynamics, we cannot derive a closed expression for the efficiency because it evolves along the chaotic attractor. Nevertheless, we can average the efficiency, which is computed as a long-time average along the attractor, since the limit set is densely covered by a typical trajectory. This leads us to write the average thermodynamic efficiency as
\begin{equation*}
\langle \eta \rangle = \frac{\sin\beta}{\beta+\sin\beta}\left(1-\frac{K}{q_1}\,\langle b_1\rangle\right),
\end{equation*}
where the average of a dynamical variable $g(t)$ is computed for non-periodic trajectories as
\begin{equation*}
\langle g(t) \rangle = \lim_{\tau\to\infty}\frac{1}{\tau}\int_{0}^{\tau} g(t)\,\mathrm{d}t,
\end{equation*}
which remains bounded for strange attractors. The general trend of thermodynamic efficiency is plotted in Fig.~2. Except for the bifurcation \emph{critical points} leading to chaotic dynamics, there is a general tendency of efficiency to increase, even beyond the chaotic regime, where periodic windows abound \cite{Lopez2023Lorenz}.

\subsection{Entropy generation and exergy destruction}

Assuming thermal equilibrium with the environment at temperature $T$ and using the equivalence between frictional work and heat \cite{Lopez2023Lorenz}, the entropy exchange rate is $-\nu\omega^2/T$.
In a stationary state, the second law yields the entropy-generation rate
\begin{equation}
\dot S_{\mathrm{gen}}=\frac{\nu \omega^2}{T}\ge 0.
\label{eq:Sgen_single}
\end{equation}
Exergy destruction is directly related to entropy production \cite{Lopez2023Lorenz}, and satisfies $\dot X_{des} = T \dot S_{gen}$. Again, in regime (ii), we can obtain the closed analytical formula
\begin{equation*}
\omega^{*}=K\sqrt{\rho-1},
\end{equation*}
which yields
\begin{equation*}
\dot S_{\mathrm{gen}}^{*}=\frac{\nu K^{2}}{T}\,(\rho-1),
\end{equation*}
as long as the Rayleigh number $\rho$ is greater than one and smaller than the critical value $\rho_c$, i.e., before the onset of chaos. Again, the general trend of entropy generation is plotted in Fig.~2 for the different regimes, along with the thermodynamic efficiency. As the reader can verify, the two trends are in direct relation of proportionality, so that more efficient machines produce more work output in proportion to their energetic income. However, since such energy production is ultimately dissipated into the environment yielding uniform dynamics, this entails more entropy generation, as well \cite{Lopez2023Lorenz}.
\begin{figure}
\centering
\includegraphics[width=0.98\linewidth]{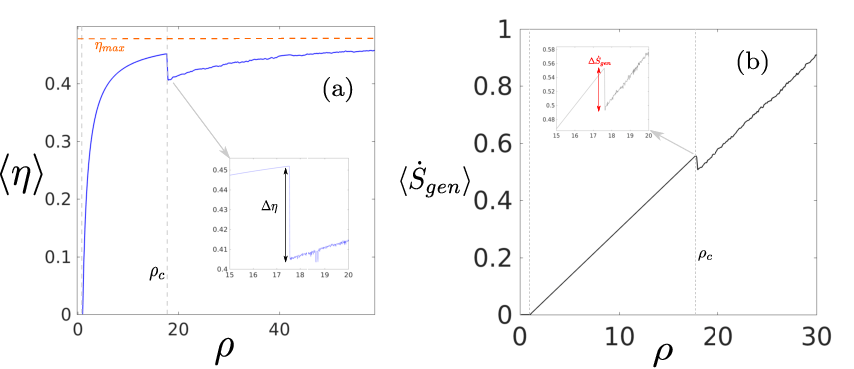}
\caption{The values of the (a) average efficiency $\langle \eta  \rangle$ of the Malkus-Lorenz waterwheel, and its (b) average entropy production $\langle \dot{S}_{gen}  \rangle$, as the input power of water $q$ is increased (and hence $\rho$). The parameters $I=1~kg \cdot m^2$, $g_{eff} = g \sin \alpha =10/\pi ~m \cdot s^{-2}$, $\nu=10~s^{-1}$, $K=1~s^{-1}$ have been set. As we can see, both magnitudes are positively correlated and tend to increase monotonically as the system is increasingly driven far from equilibrium, except for the first transition to chaotic dynamics.}
\label{fig:fig2}
\end{figure}

\section{Coupled dissipative structures: series and parallel couplings}

We now introduce two canonical couplings for the double-wheel, yielding two-engine assemblies. Despite its simplicity, we expect the present analysis to state fundamental thermodynamic laws, which might later be used to obtain equivalent thermodynamic dissipation laws for reduced complex networks, as it is commonly done in the theory of circuits \cite{Johnson2003}.

In both cases, each wheel has its own state variables $(a_i,b_i,\omega_i)$, flow rates $q_i$, and leakage parameters $K_i$. On the other hand, and for simplicity, in the present work, we assume that the wheels are identical in their structure $(r,\alpha,I)$ and rate of dissipation $\nu$. In this manner, asymmetries in the system are restricted to income and outcome flows of matter of each unit. Nevertheless, our conclusions are easily extensible to heterogeneous systems of wheels, since the theorems demonstrated below are independent on the specific dynamical system comprising the units. The coupling is implemented at the level of the Fourier-mode equations, consistent with the waterwheel derivation.

\subsection{Series coupling in cascade}

For series coupling, we can think of the wheels forming a tower with two levels, the second wheel being driven by the outflow or leakage of the first wheel, and nothing else. From an ecological point of view, if we think of a species population, this can be interpreted as a compartment feeding another compartment at a lower hierarchy level with its detritus. Alternatively, since the units of a species can die off and feed another species, we can also use the present coupling scheme to represent the predation of a certain level of the hierarchy, by a species at a higher level. Assuming no loss of water in-between, the rate of matter loss by the wheel above $K\mu_1(\theta)$ is exactly the water gained by the second wheel below. Thus, using Fourier analysis again, a representative reduced model reads
\begin{align}
\dot a_1(t) &= \omega_1(t) b_1(t) - K_1 a_1(t),\\
\dot b_1(t) &= -\omega_1(t) a_1(t) - K_1 b_1(t) + q_{1,1},\\
\dot\omega_1(t) &= -\frac{\nu}{I}\omega_1(t) + \frac{\pi g r\sin\alpha}{I}a_1(t),\\
\dot a_2(t) &= \omega_2 b_2(t) - K_2 a_2(t) + K_1 a_1(t),\\
\dot b_2(t) &= -\omega_2 a_2(t) - K_2 b_2(t) + K_1 b_1(t),\\
\dot\omega_2(t) &= -\frac{\nu}{I}\omega_2(t) + \frac{\pi g r\sin\alpha}{I}a_2(t),
\label{eq:series_ode}
\end{align}
where the primary system (master) here draws energy directly from the environment, feeds its cycle performing work, and pours its degraded outcome of matter to the second system (slave), which in turn uses this flow as an income of energy to produce its own cyclic work output, finally leaking its waste to the environment. 

These equations implement a minimal directed chain of two units connected in a cascade. It will be demonstrated in the following sections using the first law of thermodynamics that, since the intermediate flow is both leaked and absorbed within the total system, the total efficiency of the coupled structure is given by the income of the first wheel and the leakage of the second. Moreover, for arrangements in series, both wheels are producing work; the energetic distance between the material source of water and the reservoir to which the potential energy is released is always greater, and therefore, the efficiency of a series of wheels will always increase the overall thermodynamic efficiency, what favors the formation of hierarchies in nature.
\begin{figure}
\centering
\includegraphics[width=0.92\linewidth]{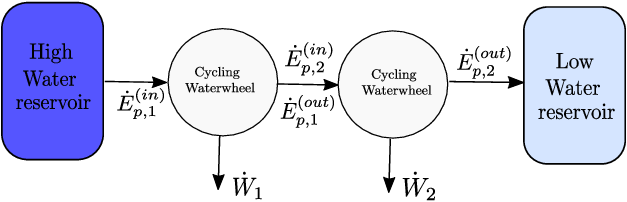}
\caption{Two Malkus-Lorenz waterwheels are coupled in series (master-slave), so that the waste of the first wheel, which is connected to the high water reservoir (master) is incorporated to the cycle of a second waterwheel (slave), which is connected to a low water potential energy reservoir.}
\label{fig:sw}
\end{figure}

In other words, by connecting machines in series, an increase in the external gradient between the reservoirs is attained. This is particularly relevant to the Earth's dynamics, since ecosystem development allows the Earth to stay cold, increasing the gradient with respect to the sun. From the point of view of thermal engines, the heat baths of the engine can be displaced dynamically, rendering the system more efficient. This contrasts with conventional non-evolving machines, where back-reaction from the fixed environment on the systems efficiency can be generally disregarded. Indeed, artificial non-evolving machines generally work with fixed thermal baths. When the environmental constraints are fixed, the only way to improve thermodynamic efficiency
is by reducing the irreversibilities, approaching the Carnot limit. However, in natural ecosystems machines couple and form complex networks, modifying the properties and limits of their environmental reservoirs, yielding a permanent feedback between species evolution and the reshape of the environment.

\subsection{Parallel coupling with diffusive exchange}

As a second fundamental arrangement, we consider a symmetric diffusion-like coupling of strength $c$ between the two engines. This exchange of water represents two dissipative structures that can be fed from an environmental common energy source and release their waste to the environment at a specific ``trophic level''. Strictly speaking, two wheels working in rigorous parallel should not communicate between them. However, in nature it is very frequent to find elements at a particular trophic level interacting. Furthermore, by means of these interactions dynamical engines at a certain level can synchronize and, when a regime of complete synchronization is attained, their diffusive couplings vanish \cite{Eroglu2017}. Therefore, we extend the parallel concept to encompass wheels that share their excess energy budget, which can be thought of as a form of redistribution of resources (\emph{e.g.} commensalism). The diffusion of water is easily modeled by imposing that any relative excess from one wheel is translated into its connected partner. A representative reduced model for the first harmonics is given by
\begin{align}
\dot a_1(t) &= \omega_1 b_1(t) - K_1 a_1(t) + c(a_2(t)-a_1(t)),\\
\dot b_1(t) &= -\omega_1 a_1(t) - K_1 b_1(t) + c(b_2(t)-b_1(t)) + q_{1,1},\\
\dot\omega_1(t) &= -\frac{\nu}{I}\omega_1(t) + \frac{\pi g r\sin\alpha}{I}a_1(t),\\
\dot a_2(t) &= \omega_2 b_2(t) - K_2 a_2(t) + c(a_1(t)-a_2(t)),\\
\dot b_2(t) &= -\omega_2 a_2(t) - K_2 b_2(t) + c(b_1(t)-b_2(t)) + q_{2,1},\\
\dot\omega_2(t) &= -\frac{\nu}{I}\omega_2(t) + \frac{\pi g r\sin\alpha}{I}a_2(t),
\label{eq:parallel_ode}
\end{align}
together with the corresponding $n=0$ (mean) modes, if desired. This coupling aligns with interactions presented in traditional investigations on the synchronization of self-sustained oscillators \cite{Eroglu2017}, which enables a discussion on the possible benefits of periodic and chaotic synchronization on the thermodynamic efficiency of far-from-equilibrium cycling systems.

Importantly, we insist that, when the wheels are completely symmetric ($K_1=K_2$ and $q_{1,1}=q_{2,1}$) and both wheels are completely synchronized $a_1(t)=a_2(t)$, $b_1(t)=b_2(t)$, and $\omega_1(t)=\omega_2(t)$, the diffusive coupling vanishes, yielding two dissipative structures that work in strict parallel between two material reservoirs without any interaction between them. The fact that the correlations persist even when the interaction term vanishes will be crucial for our result to be valid in cases where diffusive coupling is non-zero ($c \neq 0$). If such synchronization manifold exists and  it is stable, it is evident that the double-wheel system can be reduced to a single wheel system that draws potential energy at an average rate $\bar{q}_1=(q_{1,1}+q_{2,1})/2$ and also releases it at an average rate $\bar{K}=(K_{1}+K_{2})/2$. As we shall see, the fact that coupled wheels in parallel tend to average their efficiencies is a general feature of this sort of arrangements.
\begin{figure}
\centering
\includegraphics[width=0.75\linewidth]{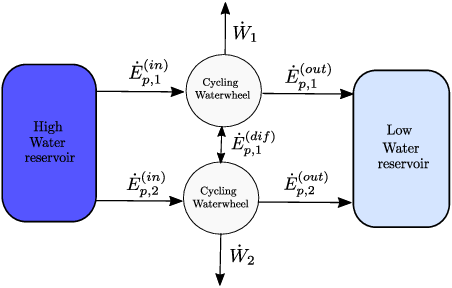}
\caption{Two Malkus-Lorenz waterwheels are coupled in parallel (diffusively), so that both wheels draw power from the high potential energy reservoir and, after partly converting it to work, they reject the degraded remaining energy to the low energy reservoir. Depending on the constant of diffusion, the waterwheels also tend to level-off their water content by transferring their relative excesses.}
\label{fig:pw}
\end{figure}

\section{Two fundamental association laws}

A complex network can be significantly reduced by synthesizing consecutive discrete elements that are combined in parallel or in series. Therefore, we now formalize the two core results as theorem statements. Consider an engine $E$ characterized by the income and outcome of energy, where the input power rate $\dot E^{(\mathrm{in})} > 0$, while the output (rejected) power rate $\dot E^{(\mathrm{out})}\ge 0$, and useful power $\dot W = \dot E^{(\mathrm{in})} - \dot E^{(\mathrm{out})}$.
Define efficiency
\begin{equation}
\eta = \frac{\dot W}{\dot E^{(\mathrm{in})}} = 1 - \frac{\dot E^{(\mathrm{out})}}{\dot E^{(\mathrm{in})}},\qquad 0\le \eta \le 1.
\label{eq:eta_abstract}
\end{equation}
This agrees with the waterwheel definition appearing in Eq.~\eqref{eq:eta_def_general}. From this definition applied to each waterwheel, as well as to the total system, we can obtain the fundamental association laws.

\subsection{Series association law}

We first consider the case of several engines connected in series, and prove that such an arrangement always increases or leaves equal the thermodynamic efficiency. This result suggests that long chains of cascading thermodynamic engines (\emph{i.e.} hierarchies with respect to primary resources) should be favored in nature, as long as downstream elements (slaves) are not dramatically deprived of energy income from upstream elements (masters). The theorem is stated as follows: 

\begin{theorem}
\label{thm:series}
Let $E_1,\dots,E_n$ be engines with efficiencies $\eta_i\in[0,1]$.
Assume a \emph{series} coupling in which the rejected stream of stage $i$ is the input of stage $i+1$:
$\dot E_{i+1}^{(\mathrm{in})}=\dot E_i^{(\mathrm{out})}$ for $i=1,\dots,n-1$.
Define the equivalent series engine by total input $\dot E^{(\mathrm{in})} = \dot E_1^{(\mathrm{in})}$ and total work $\dot W = \sum_{i=1}^n \dot W_i$.
Then the equivalent efficiency is
\begin{equation}
\eta_{\mathrm{s}}
= 1 - \prod_{i=1}^n (1-\eta_i).
\label{eq:eta_series_general}
\end{equation}

\end{theorem}

\begin{proof}
By definition, $\dot E_i^{(\mathrm{out})} = (1-\eta_i)\dot E_i^{(\mathrm{in})}$.
Under series coupling, $\dot E_{i+1}^{(\mathrm{in})} = \dot E_i^{(\mathrm{out})}$, hence
$\dot E_{i+1}^{(\mathrm{in})} = (1-\eta_i)\dot E_i^{(\mathrm{in})}$.
Iterating gives $\dot E_n^{(\mathrm{out})} = \prod_{i=1}^n(1-\eta_i)\dot E_1^{(\mathrm{in})}$.
Total work is $\dot W = \dot E_1^{(\mathrm{in})} - \dot E_n^{(\mathrm{out})}$, therefore
$\eta_{\mathrm{s}} = \dot W/\dot E_1^{(\mathrm{in})} = 1 - \prod_{i=1}^n(1-\eta_i)$.
\end{proof}

For $n\ge 2$ and $0<\eta_i<1$, we have $\eta_{\mathrm{s}}>\max(\eta_i)$.
This is the analytic equivalent engine in series law. Just as a series of resistors increase the overall resistance of an electric circuit, here we see that series connection increases the overall efficiency of a system. In particular, for two stages we get
\begin{equation}
\eta_{\mathrm{s}} = \eta_1 + \eta_2 - \eta_1\eta_2,
\label{eq:eta_series_two}
\end{equation}
which nicely illustrates that the resulting efficiency is not simply the addition of efficiencies, since thermodynamic efficiency is bounded between zero and one. A particularly simple result concerns the case where all the engines are equally efficient. In such case, and assuming that we couple $n$ total engines, we obtain $\eta_s=1-(1-\eta)^n$, which in the limit $n  \rightarrow \infty$ tends to a perfect engine. Note, however, that this by no means violates Carnot's theorem, since by introducing long cascades we are forcefully displacing the ``cold'' (``hot'') material reservoirs to zero (infinity).

\subsection{Parallel association law}

Secondly, we consider the case of several engines in parallel, and prove that such an arrangement strictly averages the thermodynamic efficiency when there is no diffusion between the parallel elements. This suggests that parallel chains of cascading thermodynamic engines should be favored in nature as well, but not because of an increase in thermodynamic efficiency, but for the fact that the total energy flow through the network is increased. This role is intimately related to autocatalysis and population growth, which entails the creation of numerous thermodynamic units within a single compartment at the same hierarchy level through geometric growth \cite{Bacaer2011}, which then can communicate by coupling and synchronizing. The convenience of the geometric growth of species in ecosystem networks is that it serves as a basis for the successive selective pressures of natural selection, sculpting the evolutionary landscape \cite{Nowak2006}. The former provides machines working in parallel, while the latter favors the most efficient cascades, both mechanisms working in a permanent feedback loop. Mathematically, our second theorem can be stated as follows:

\begin{theorem}
\label{thm:parallel}
Let $E_1,\dots,E_n$ be engines operating in \emph{parallel} under a shared supply such that the total input power splits into nonnegative portions
$\dot E^{(\mathrm{in})}=\sum_{i=1}^n \dot E_i^{(\mathrm{in})}$ with $\dot E_i^{(\mathrm{in})}>0$.
Define weights $w_i = \dot E_i^{(\mathrm{in})}/\dot E^{(\mathrm{in})}$ so that $w_i(t)>0$ and $\sum_i w_i(t)=1$.
Then the equivalent parallel efficiency satisfies
\begin{equation}
\eta_{\mathrm{p}} = \sum_{i=1}^n w_i \eta_i.
\label{eq:eta_parallel_general}
\end{equation}
\end{theorem}

\begin{proof}
Total work is $\dot W = \sum_i \dot W_i = \sum_i \eta_i \dot E_i^{(\mathrm{in})}$.
Divide by $\dot E^{(\mathrm{in})}$ to obtain
$\eta_{\mathrm{p}} = \dot W/\dot E^{(\mathrm{in})} = \sum_i (\dot E_i^{(\mathrm{in})}/\dot E^{(\mathrm{in})})\eta_i = \sum_i w_i\eta_i$.
\end{proof}

This association law holds rigorously only in the case uncoupled waterwheels ($c=0$) in parallel, or when the system has completely synchronized ($\dot{E}^{(\text{dif})}_{p,i}=0$). Concerning the Malkus-Lorenz waterwheel, the weights are given by $w_i = q_ {1,i}/(q_{1,1}+q_{1,2})$, i.e. weights are the rate inflow fractions. However, things can be far more complicated when diffusion between parallel units is allowed, and generally the weights evolve in time $w_i(t)$, yielding a probability distribution that changes dynamically, as it is customary in the study of stochastic processes. Moreover, the resulting thermodynamical efficiency must be corrected when $\dot{E}^{(\text{dif})}_{p,i} \neq 0$, since diffusion modifies the incomes and outcomes of the subunits, which define the thermodynamic efficiency of such elementary systems. For two machines that are not fully synchronized, diffusion is active and the law can be written as $\eta_{\mathrm{p}} = w_1 \eta_1+w_2 \eta_2 - (\epsilon_1 (1-\eta_1)+\epsilon_2 w_2 (1-\eta_2))$, where $\epsilon_1=\dot{E}^{(\text{dif})}_{p,1}/(\dot E^{(\mathrm{in})}_{p,1}+\dot{E}^{(\text{dif})}_{p,1})$ and $\epsilon_2=\dot{E}^{(\text{dif})}_{p,1}/(\dot E^{(\mathrm{out})}_{p,2}+\dot{E}^{(\text{dif})}_{p,1})$. We stress that when units are fully synchronized there is no diffusion at all $\epsilon_i=0$, restoring the result of the previous theorem. We shall explore numerically the effects of diffusive coupling and synchronization on the total thermodynamic efficiency of the system. Note that, since $0 \leq \epsilon_i \leq 1$, according to the previous formula for $\eta_{\mathrm{p}}$, all the rest being equal, diffusive processes tend to reduce thermodynamic efficiency. Therefore, synchronization, when possible, should favor thermodynamic efficiency, since it makes $\epsilon_i$ tend to zero, rendering the system diffusionless.

Summarizing, we insist that the connection of dissipative structures in parallel does not increase efficiency in exergy degradation, but it rather increases the overall energy flow through the ecosystem \cite{Odum1995}, while interpolating the thermodynamic efficiency between the efficiencies of the two parallel subsystems. 

\section{Application to coupled Lorenz systems}

We now connect the abstract results derived in the previous section to explicit water wheels driven by gravitational potential-power terms and connected through water flows, extending the single-wheel calculation of Ref.~\cite{Lopez2023Lorenz} to coupled systems. There exist numerous studies involving networks of Lorenz systems \cite{Mat97,Baya2024}, and since the Lorenz system can be considered a universal model for chaotic and heterochaotic dynamics \cite{Lorenz1963,Saiki2018}, its nonlinear features make it an ideal candidate for constructing universal arbitrarily complex networks in order to examine the effects of network topology and dynamics on thermodynamic efficiency.

\subsection{Potential-power bookkeeping for two wheels}

For wheel $E_i$, define the gravitational potential energy (up to constants) using the same height function $h(\theta)=r\sin\alpha(1+\cos\theta)$ used in the single-wheel derivation \cite{Lopez2023Lorenz}. Using the Fourier coefficients of the mass density, one obtains for the potential energy output
\begin{equation}
\dot{E}^{(\mathrm{out})}_{p,i} = \pi g r\sin\alpha\,K_i\,(2b_{0i} + b_i),
\label{eq:Epout_i}
\end{equation}
in direct analogy with the single wheel. For symmetric injection (hose sector width $\beta$), the inflow power reads
\begin{equation}
\dot{E}^{(\mathrm{in})}_{p,1} = 2g r\sin\alpha\, q_{1}(\beta+\sin\beta),
\label{eq:Epin_i}
\end{equation}
again matching the single-wheel result. The previous two equations allow to compute the efficiency of series coupling, since $\dot{E}^{(\mathrm{in})}_{p,2}=\dot{E}^{(\mathrm{out})}_{p,1}$. 

Next, for each wheel, we retain the structural definition
\begin{equation}
\eta_i = 1 - \frac{\dot{E}_{p,i}^{\mathrm{(out)}}}{\dot{E}_{p,i}^{\mathrm{(in)}}},
\label{eq:eta_i_coupled_def}
\end{equation}
where $\dot{E}_{p,i}^{\mathrm{(in)}}$ and $\dot{E}_{p,i}^{\mathrm{(out)}}$ are as defined in the previous paragraph, plus any coupling-induced effective power entering wheel $i$, as in the case of diffusion. Importantly, when the dynamics of the wheels is chaotic, we have to average the flows of energy, and since now both the income and outcome can evolve in time due to diffusion, we must compute the efficiency using the equation
\begin{equation}
\langle  \eta_i \rangle \equiv 1 - \frac{\langle \dot{E}_{p,i}^{\mathrm{(out)}}\rangle}{\langle\dot{E}_{p,i}^{\mathrm{(in)}}\rangle},
\label{eq:eta_i_coupled_def}
\end{equation}
where the average is carried as before. We now provide explicit formulas for the series and parallel couplings, which will be used in the following sections to perform simulations.

For series coupling, the computations provide the corresponding per-stage expressions, with the downstream stage’s effective inflow determined by the upstream outflow. The formulas for the efficiency of the two units and the total efficiency are given by Eqs.~\eqref{eq:Epout_i} and \eqref{eq:Epin_i}. At the abstract level, the global series efficiency is governed by Theorem~\ref{thm:series}, consistent with the corresponding association law. These results are numerically confirmed in the following sections. We get for the efficiency of each machine
\begin{align}
\eta_{1}(t) &=1-\dfrac{\beta+K_1(\sin\beta/q_{1,1}) b_1(t)}{\beta+\sin\beta}, \\
\eta_{2}(t) &=1-\dfrac{\beta+K_2(\sin\beta/q_{1,1}) b_2(t)}{\beta+K_1(\sin\beta/q_{1,1}) b_1(t)}.
\label{eq:Epdif}
\end{align}
Finally, the total efficiency of the system can be computed as
\begin{align}
\eta(t) &=1-\dfrac{\beta+K_2(\sin\beta/q_{1,1}) b_2(t)}{\beta+\sin\beta}.
\label{eq:Epdif}
\end{align}

It is immediate to check that $\eta=\eta_1 + \eta_2 - \eta_1 \eta_2$, as we shall verify using numerical technique in the following sections.

Concerning parallel arrangements, from the latter Eqs.~\eqref{eq:Epout_i} and \eqref{eq:Epin_i}, it is immediate to derive the weights for the uncoupled wheels in parallel ($c=0$), as well. In the case of parallel engines under diffusive coupling ($c \neq 0$), we must also compute wheel 1 exchange of mass (and hence potential power) with wheel 2, and vice-versa. The computations derive exchange contributions
\begin{align}
\dot E_{p,1}^{\mathrm{(dif)}} &= \pi g r\sin\alpha\,c\left(2(b_{02}-b_{01}) + (b_2-b_1)\right),\\
\dot E_{p,2}^{\mathrm{(dif)}} &= \pi g r\sin\alpha\,c\left(2(b_{01}-b_{02}) + (b_1-b_2)\right).
\label{eq:Epdif}
\end{align}
and, therefore, $\dot E_{p,1}^{\mathrm{(dif)}} = -\dot E_{p,2}^{\mathrm{(dif)}}$. These terms vanish when the two engines share identical distributions, and they encode a physically transparent notion: diffusive coupling transports potential power from wheels having higher budget to lower loaded wheels, and thereby contributes to (or subtracts from) each wheel’s effective balance. Thus it is crucial to note that, given a particular wheel, the diffusive term can form part of its income or its outcome, depending on the relative budget as compared to neighboring wheels. This fact must be carefully taken into account when computing the thermodynamic efficiency. Thus for parallel symmetric diffusive coupling with diffusion coefficient $c$, the computations yield simplified expressions of the form 
\begin{align}
\eta_1(t) &= 1 - \frac{\pi K_1 b_{01}/q_1 + K_1 (\sin\beta/q_{1,1})\,b_1(t) + c \pi (b_{10}-b_{20})/q_1 + c(\sin\beta/q_{1,1})(b_1(t)-b_2(t))}{\beta+\sin\beta},
\label{eq:eta_parallel_symmetric_1} \\
\eta_2(t) &= 1 - \frac{\pi K_2 b_{02}/q_2 + K_2(\sin\beta/q_{2,1})\,b_2(t)}{\beta+\sin\beta +  c \pi (b_{10}-b_{20})/q_2 + c(\sin\beta/q_{2,1})(b_1(t)-b_2(t))} ,
\label{eq:eta_parallel_symmetric_2}
\end{align}
and the global efficiency coincides the weighted average $\eta = w_1\eta_1 + w_2\eta_2$ as long as $c$ is sufficiently hight and synchronization is attained. Note that we have assumed here that $b_1(t)>b_2(t)$, and therefore, the diffusion in this particular case goes from the first wheel to the second. The values of $b_{10}$ and $b_{20}$ are computed from the equilibrium values of the zeroth order differential equations resulting from the Fourier expansion. We also have that $q_{i,1}=2q_i \sin\beta/\pi$. When the sign changes to the opposite, these definitions must be changed by putting the input of diffusion in the first wheel, in the denominator of Eq.~\eqref{eq:eta_parallel_symmetric_1}. Importantly, we note that when the wheels are not synchronized $(\epsilon_i \neq 0)$, the diffusive terms are non-null $(\dot E_{p,1}^{\mathrm{(dif)}} \neq 0)$, and the total efficiency is reduced according to $\eta_{\mathrm{p}} = w_1 \eta_1+w_2 \eta_2 - (\epsilon_1 (1-\eta_1)+\epsilon_2 w_2 (1-\eta_2))$. This lack of synchronization can be considered a consequence of one of two facts: either the wheels are not symmetric concerning leakage or gain from the reservoirs (or other internal parameters), or the diffusive coupling is insufficient.					

\section{Numerical Results}

This section summarizes the key numerical and phenomenological findings reported in the presentation, stated here as observations to guide future systematic studies. All the integrations have been performed MATLAB's Dormand-Price adaptive Runge-Kutta method. Again, we address the numerical results, distinguishing the two fundamental arrangements: master-slave coupling and diffusive coupling. 
\begin{figure}
\centering
\includegraphics[width=0.48\linewidth]{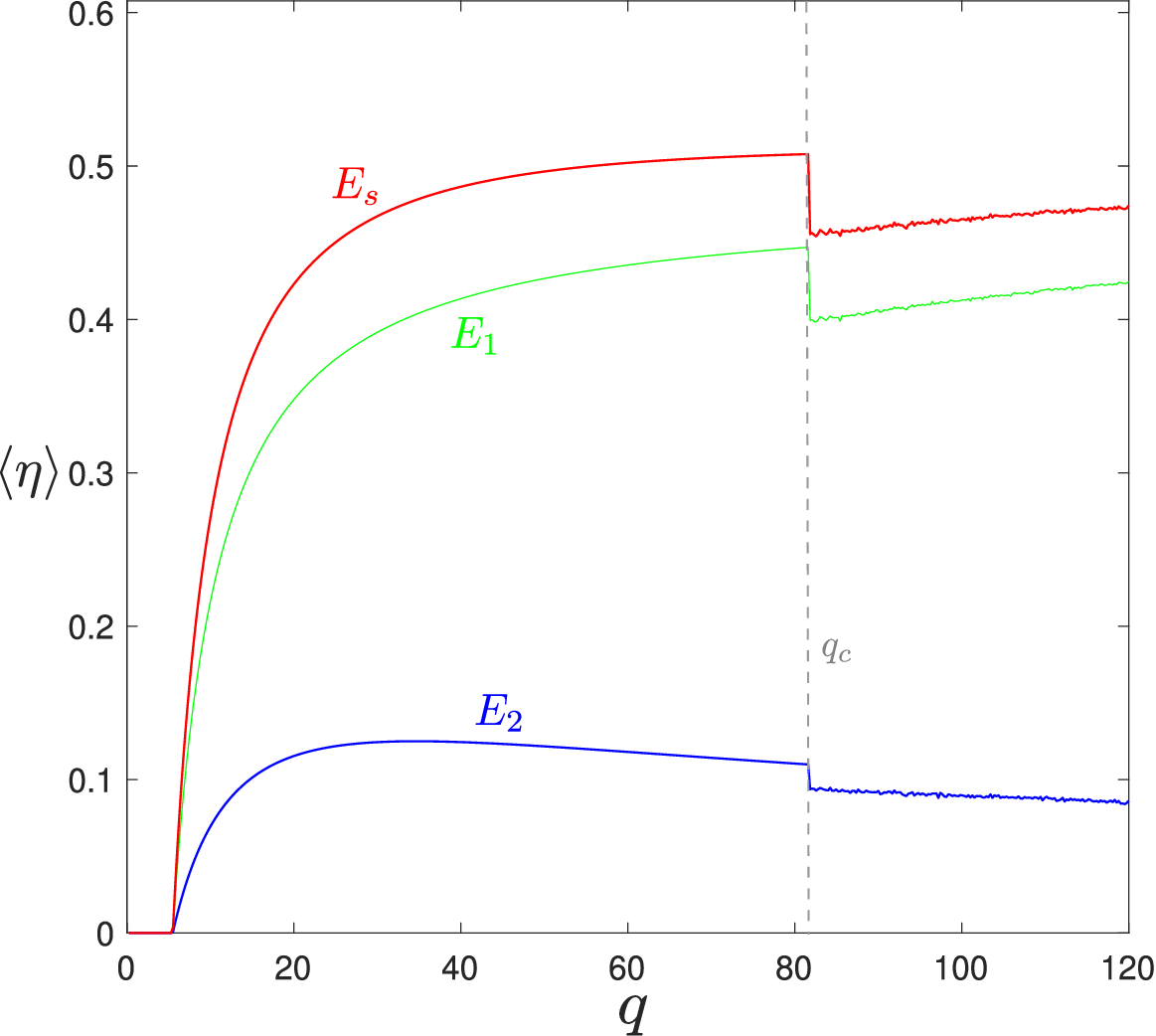}
\caption{Efficiency landscape under series (master--slave) coupling as a function of the rate of energy draw from the high potential energy/matter reservoir by the master system. The reported phenomenology suggests that the law $\langle \eta_s \rangle=\langle\eta_1\rangle+\langle\eta_2\rangle-\langle\eta_1 \rangle \langle\eta_2\rangle$ holds nicely. As we can see, the efficiency of the equivalent engine in series $E_s$ is always above $E_1$ and $E_2$, thus master-slave coupling enhance thermodynamic efficiency.}
\label{fig:fig5}
\end{figure}

\subsection{Series coupling}

For master--slave series coupling, numerical exploration reveals that, as we connect elements in series, we can indeed replace the chain for an equivalent machine whose efficiency is always
greater (see Fig.~\ref{fig:fig5}). This is done by displacing the energy gradient between the primary reservoir and the lowest energy reservoir, what suggests that the reservoirs to which energy is delivered and from which energy is drawn, can evolve and distance from each other. This contrast with conventional engines with fixed energy reservoirs, which can only increase their efficiency by becoming increasingly reversible, and approaching Carnot's efficiency. 

The theoretical results derived in the previous section hold nicely, where the curves resulting from the direct definition of $\langle \eta_s \rangle$ match to very high accuracy to the results attained by its computation from $\langle \eta_1 \rangle$ and $\langle \eta_2 \rangle$. The chosen parameters are $\nu=~10 s^{-1}$, $I=~1 kg ~\cdot m^{2}$, $g_{eff}=10/\pi~m \cdot s^{-2}$, $K_1=0.5~s^{-1}$ and $K_2=1.5~s^{-1}$, so that the engine $E_1$ has greater efficiency than the second wheel $E_2$. This is also a consequence of the fact that the deliver of power of the second engine depends on the release of the former, which is smaller than the loss by the rate of its losses $K_2$. A new important result as compared for these parameter setting is that the efficiency of the second wheel downstream does not increase monotonically with the input rate of water of the system $q$, which contradicts statements made in Ref.~\cite{SchneiderKay1994}. This is a consequence of the fact that the first wheel is taking advantage of this power by increasing efficiency. The rate of power loss from the master wheel depends on its dynamics, and therefore we do not need to obtain a monotonically increasing average efficiency with $q$ for elements downstream, as compared with $E_1$, whose efficiency increases monotonically as the system is driven far from equilibrium. Nevertheless, the entire system does increase its efficiency as it is pushed further from equilibrium. Interestingly, the transitions to chaotic dynamics occur at once for the cascade at $q_c=82$, with sudden drops in the thermodynamic efficiency for both wheels.
\begin{figure}[h]
\centering
\includegraphics[width=0.48\linewidth]{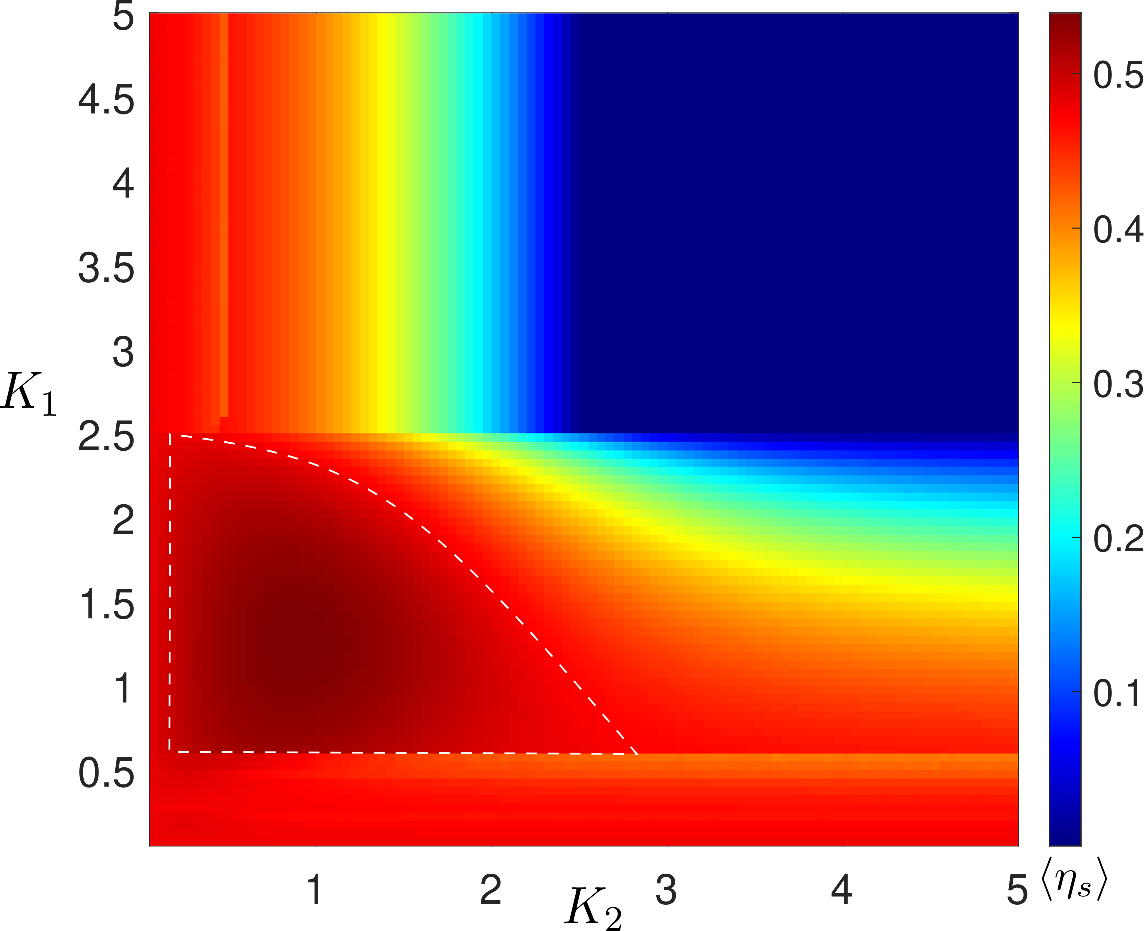}
\caption{Efficiency landscape under series (master--slave) coupling in the loss-rate parameter space $(K_1,K_2)$ of the two wheels. We find a region (dashed gray line) for which the efficiency is maximized.}
\label{fig:fig6}
\end{figure}

Numerical exploration also indicates the existence of a region of maximal efficiency in the $(K_1,K_2)$ (loss-rate) parameter space: too large upstream losses (high $K_1$) at the master engine $E_1$ reduce the available upstream budget losing opportunity to destroy exergy, while too-slow downstream capture (low $K_1$) reduces the effective flow at engine $E_2$, compromising global efficiency. Therefore, an intermediate value of $K_1$ seems to capture the highest thermodynamic efficiency, by considering the counterbalancing of these two effects. In other words, in the master-slave coupling, the master wheel must not keep all the flux for itself, since by converting its input of energy to work too efficiently, the slave wheel becomes highly inefficient, since little energy is left for the element underneath. This compromises the overall efficiency of the system, rendering a less efficient dissipative structure as a whole. 

This effect is more apparent as long as the losses of the second wheel are not too high, i.e., as long as the slave wheel is efficient in itself, given a certain rate of energy income. Indeed, there is no purpose on rejecting water more efficiently to the slave wheel if it is very inefficient  (e.g. by making $K_2$ very high). Summarizing, master-slave coupling benefits overall efficiency, but this effect is more noticeable when the master leaves a sufficient amount of energy for the slave, as long as this slave is efficient too.

\subsection{Parallel coupling and synchronization}

As we connect elements in parallel, we can replace the system of engines by an equivalent machine whose efficiency is the average of the previous. However, the total power flow by the system is considerably enhanced, thus parallel arrangements increase total energy flow. This effect can be clearly appreciated in Fig.~\ref{fig:fig7}. The set of parameter values chosen for these simulations are the same as before, except for $K_1=1~s^{-1}$ and $K_2=2~s^{-1}$,  and diffusion has been neglected in the first place ($c=0$). In all our simulations we consider equal rates of water inflow $q_{1,1}=q_{2,1}$. It is clearly appreciated that the resulting efficiency is an average of the efficiency of the two parallel subsystems. Since we have considered two wheels with different losses $K_1$ and $K_2$, which draw energy at the same rate given by $q$, the transitions to chaos take place at different values of $q$. Again, these asymmetry yields interesting scenarios, where the average efficiency can experience a drop due to partial transition to chaos of one of the subsystems, irrespective of the dynamics of the other parallel element. We have ran simulations for $c>0$ and the effects of diffusion is to redistribute the power, so that the individual efficiencies of the two wheels become closer to each other (see below), making the system more equalitarian and equally or more efficient, depending on the values of $q$.

It is interesting in this case to study the possible benefits of periodic and chaotic synchronization on thermodynamic efficiency. The numerical study posed the question of whether strong coupling and synchronization improves thermodynamic efficiency and reported that synchronization is mostly neutral or beneficial. However, strong coupling can be disadvantageous for small regions involving high energy loss of the second wheel (relative to the loss of the first wheel) and when the energy capture is small. In this case the non-transmitted flow along one path (the more efficient of the two) can outperform diversification.
\begin{figure}
\centering
\includegraphics[width=0.48\linewidth]{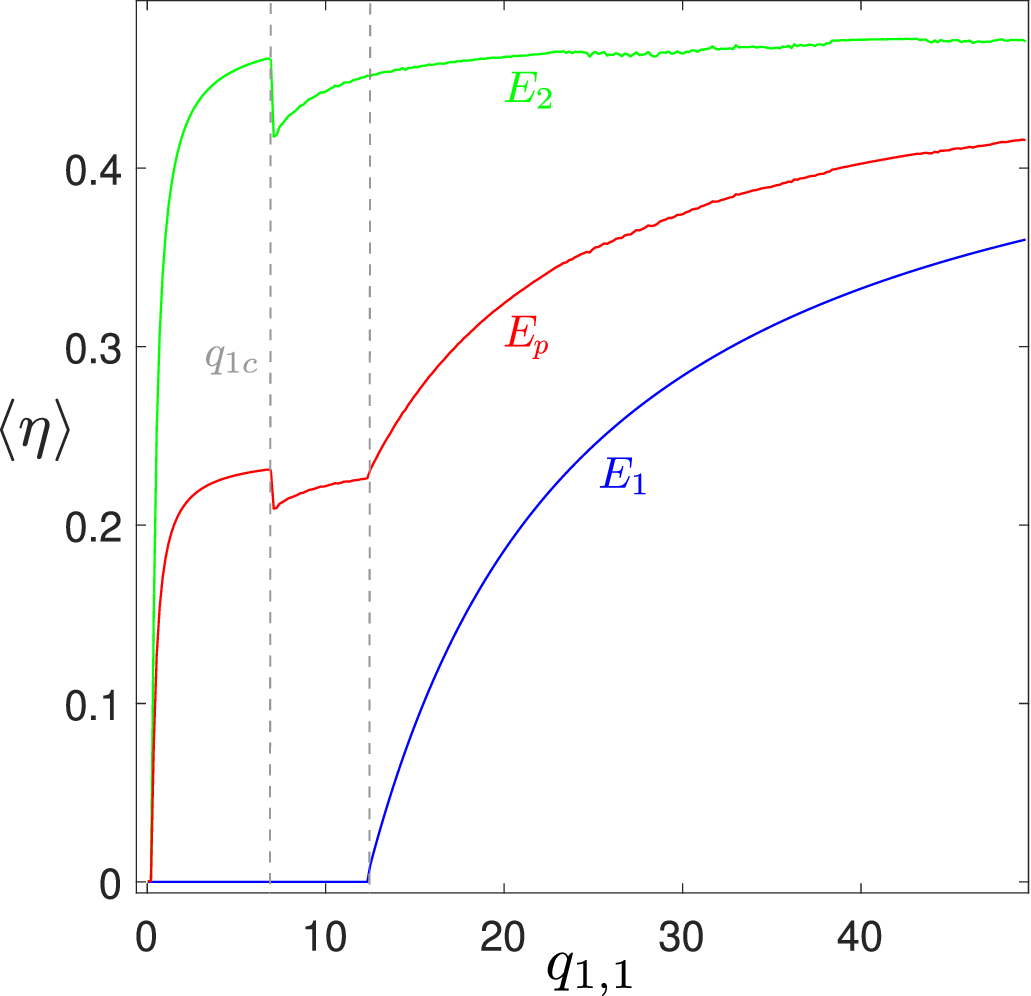}
\caption{Efficiency landscape under parallel (diffusive) coupling as a function of the rate of energy draw from the high potential energy reservoir of both elements ($c=0$). The reported phenomenology suggests that the thermodynamic efficiency averages for parallel elements. As we can see, the efficiency of the equivalent engine in parallel $E_p$ is always between $E_1$ and $E_2$.}
\label{fig:fig7}
\end{figure}

For the purpose of explaining these results, we have computed the efficiency of the whole system when $c=0$ (uncoupled), and when $c=20~s^{-1}$ (strongly coupled wheels). We have settled the value $K_1=0.5~s^{-1}$, as before. Then, we compute the difference in efficiency $\Delta \langle \eta_{\mathrm{p}} \rangle $ between the latter and the former. The particular domain of the regions that benefit from strong coupling is far form trivial (see Fig.~\ref{fig:fig8}). But since only the black and white tusk for high values values of $K_2$ is truly inefficient (here the great asymmetry in the parameters prevents synchronization), as a general rule of thumb we can say that, except when there is a very inefficient system (i.e. small values of $q$ and high values of $K_2$) being coupled to a very efficient one (i.e. high values of $q$ and small values of $K_1$), strong coupling and subsequent synchronization seems to be generally beneficial or innocuous. When a wheel is very efficient as compared to another, coupling is harmful for the overall system, since the redistribution of matter is lost by being diverted to an inefficient pathway.

\subsection{Entropy-generation curvature under coupling}

Concerning entropy generation, coupling modifies the curvature of entropy-generation trends: diffusive coupling increases entropy-generation curvature (consistent with diffusion being entropically enhancing), whereas series coupling can produce opposite curvature effects by making downstream stages dependent on upstream waste streams.
\begin{figure}
\centering
\includegraphics[width=0.48\linewidth]{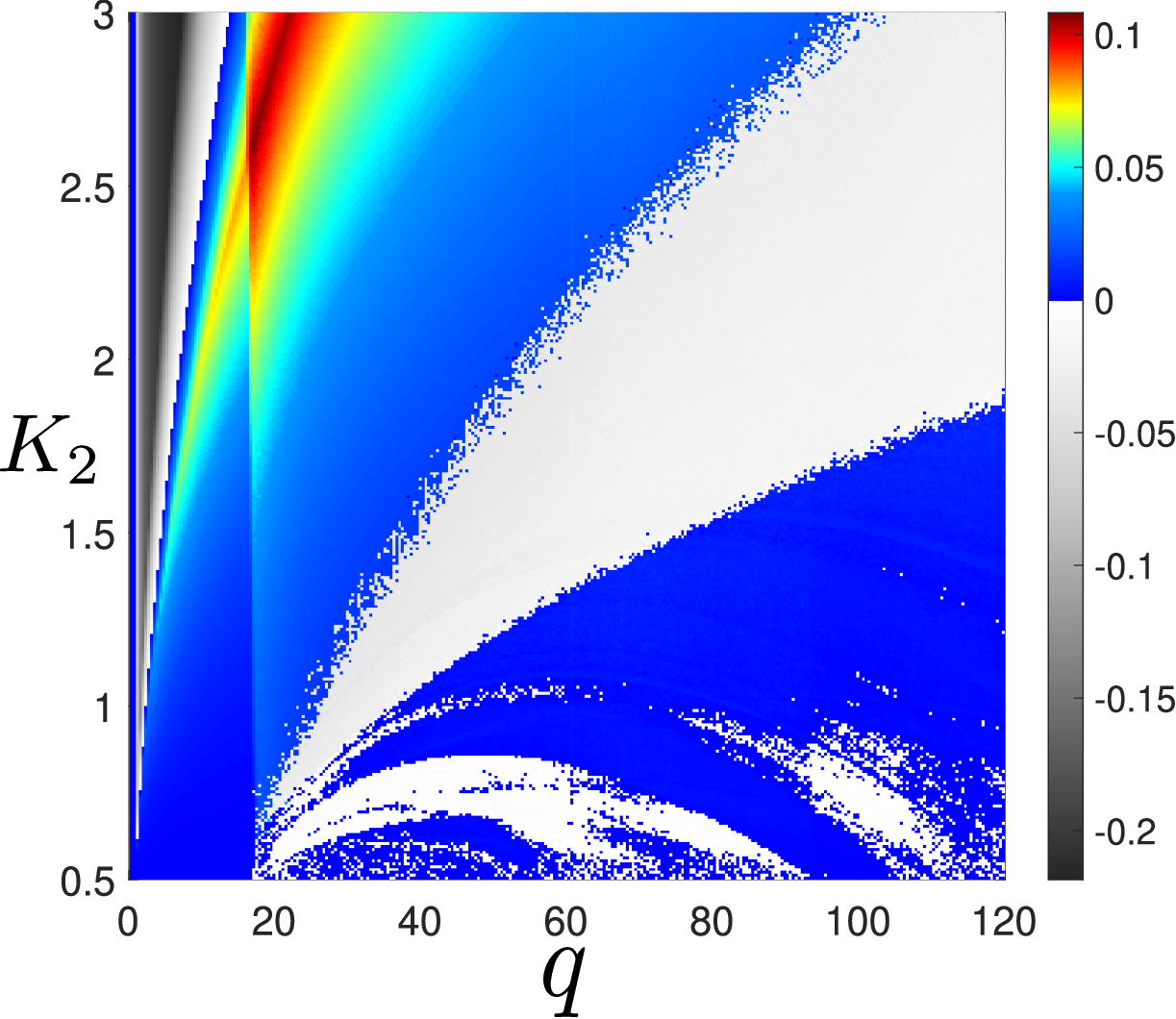}
\caption{Efficiency difference $\Delta \langle \eta_{\mathrm{p}} \rangle $ under parallel (diffusive) coupling in the gain-loss parameter space $(q,K_2)$ of the second wheel. The color bar represents the difference in efficiency for strongly coupled wheels ($c=20$) and for uncoupled wheels ($c=0$). We find a colored region for which the efficiency is enhanced by synchronization. The white regions involve parameter space domains where synchronization is mostly irrelevant, and only for very high $K_2$ and very small $q$ coupling is really harmful, corresponding to very asymmetric wheels.}
\label{fig:fig8}
\end{figure}
The results are depicted in Fig.~\ref{fig:fig9}, where we see that the slave wheel in series coupling shows concave curvature due to downstream dependence on upstream rejected flows. This contrasts to a simple wheel and also to the master wheel, where there is no curvature in entropy trends as the system is posed far from equilibrium \cite{Lopez2023Lorenz}. Moreover, for the salve wheel we appreciate that the rate of entropy loss at the transition to chaos is enhanced across the Hopf bifurcation, which is a completely new feature, due to the non-constant rate of income of matter coming from the master wheel, which suddenly becomes inefficient due to chaotic dynamics. On the other hand, for parallel wheels we appreciate that the curvature can be positive (convex) as a consequence of diffusion processes, which enhance entropy generation. This again contrasts with non-coupled wheels, that have linear increases as the system is posed far from thermodynamic equilibrium \cite{Lopez2023Lorenz}.

\section{Conclusion and discussion}

We extended the thermodynamic-efficiency framework of far-from-equilibrium engines to coupled dissipative structures. We proved two fundamental association laws: an additive series law and a weighted-average parallel law. We then specialized these theorems to coupled Lorenz engines using consistent potential-power bookkeeping, including explicit diffusive exchange contributions. Finally, we summarized numerical phenomenology suggesting maximal-efficiency regions under series coupling, mostly neutral-beneficial effects of synchronization, and coupling-dependent curvature changes in entropy generation. Together, these results provide a mathematically transparent bridge from single dissipative engines to generalized energy-flow networks, with prospective application to ecosystem-scale thermodynamic efficiency.

The single Lorenz engine already illustrates that thermodynamic efficiency can behave as a dynamical observable, generally increasing with driving yet suffering sharp discontinuities at bifurcation-driven crises in the attractor geometry. The coupled results developed here clarify how \emph{topology} alone imposes strict algebraic laws on efficiency in idealized series and parallel assemblies (Theorems~\ref{thm:series}--\ref{thm:parallel}), while the \emph{dynamics} determines the realized values of each stage’s efficiency and the stationary distribution of flows in more complex graphs \cite{Raux2024, DalCengio2023}.

The series law represented by Eq.~\eqref{eq:eta_series_general} implies that \emph{cascades} can increase global efficiency by relocating where the effective gradient is reduced. In ecological language, this mirrors the intuition that chains of consumers can exploit energy gradients staged across multiple processes. In contrast, the parallel organization trades efficiency for enhanced total throughput, consistent with the ``narrative of the maximum power'' in ecology and complex systems, as supported by H. T. Odum \cite{Odum1995}. The importance of productive chains has also been addressed in the study of economic networks \cite{carvalho2014}.
\begin{figure}
\centering
\includegraphics[width=0.92\linewidth]{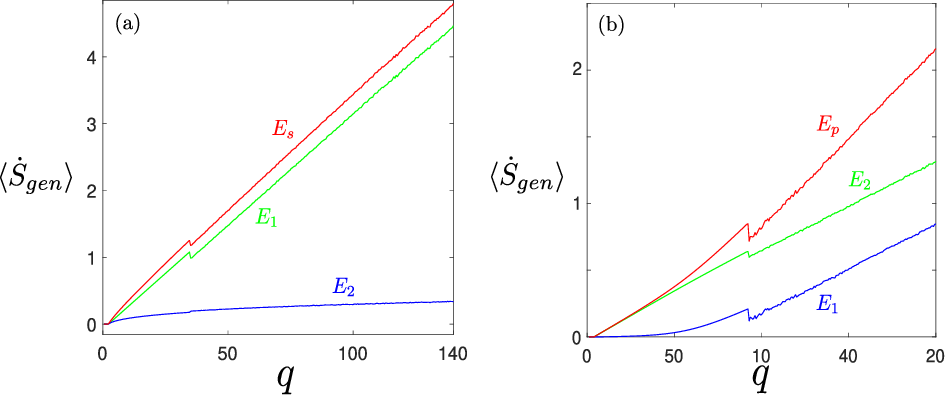}
\caption{Curvature changes in entropy generation trends under series and parallel coupling. Diffusive (exchange) coupling (b) tends to increase entropy-production curvature, while series coupling (a) can induce an opposite curvature due to downstream dependence on upstream rejected flows. Moreover, in this case the rate of entropy loss at the transition to chaos for the slave wheel is enhanced across the Hopf bifurcation, which is a completely new feature, due to the non-constant rate of income of matter coming from the master wheel, which suddenly becomes inefficient due to chaotic dynamics.}
\label{fig:fig9}
\end{figure}

The theorems~\ref{thm:series}--\ref{thm:parallel} suggest a mathematically transparent route to efficiency on networks. Let $G=(V,E)$ be a directed graph whose nodes are engines.
Let external sources inject power into a subset $V_{\mathrm{in}}\subset V$ and sinks remove rejected power from $V_{\mathrm{out}}\subset V$.
A network state defines edge flows $\dot E_{e}$ (power rates) consistent with local conservation:
incoming to node $v$ equals outgoing rejected flow plus useful work at $v$. Define global efficiency as
\begin{equation}
\eta_{\mathrm{n}} = \frac{\sum_{v\in V} \dot W_v}{\sum_{s\in V_{\mathrm{in}}} \dot E_s^{(\mathrm{in})}}.
\label{eq:eta_network_def}
\end{equation}
For \emph{series--parallel reducible} networks, repeated application of Theorems~\ref{thm:series}--\ref{thm:parallel} yields closed-form expressions. For general graphs that are irreducible, Eq.~\eqref{eq:eta_network_def} remains well-defined. However, given that our units are nonlinear dynamical systems, it is unclear if there exist association theorems capable of reducing the thermodynamic efficiency of the network to as simple scheme in terms of fundamental efficiencies, as it is the case of linear electric circuits, where the Thevenin and the Norton equivalents are frequently used to simplify circuits with complicated cycles of resistors and sources \cite{Johnson2003}.

In the long run, the natural mathematical generalization is to treat ecosystems as directed flow networks with node-level dissipation and work-like functional outputs, where efficiency is computed by Eq.~\eqref{eq:eta_network_def}. This opens two complementary research directions:

\begin{itemize}

\item \textbf{Graph-theoretic reduction and bounds.} Identify network classes where exact reduction to series--parallel blocks is possible and derive topology-dependent bounds for $\eta_{\mathrm{n}}$.

\item \textbf{Dynamical flow selection.} Couple the network efficiency definition to the dynamics that selects stationary flows (including synchronization and bifurcation structure), connecting thermodynamic performance to nonlinear stability and local activity \cite{MainzerChua2013,Arenas2008}.
\end{itemize}

The present work offers the fundamental tools and line of reasoning to accomplish these two goals in the future.

\begin{acknowledgments}
The authors thank professor Ruben Capeans for discussions on thermodynamic efficiency and coupled dissipative structures.
\end{acknowledgments}

\section*{Author Contributions}
Conceptualization \'A.G.L.; formal analysis, \'A.G.L., A.D.-B.; investigation, all authors; writing---original draft preparation, \'A.G.L.; writing---review and editing, all authors; supervision, A.D.-B. All authors have read and agreed to the published version of the manuscript.

\section*{Funding}
This work was partially supported by Spain’s Agencia Estatal de Investigación (ref. PID2024-158181NB-I00 NISA), funded by MCIN/AEI/10.13039/501100011033 and by the ERDF (“A way of making Europe”).

%\section*{Data Availability}
%No new data were created or analyzed in this study.

%\section*{Conflict of Interest}
%The authors declare no conflict of interest.

\end{document}